\documentclass[12pt, onecolumn]{IEEEtran} 
\usepackage{subfigure, epsfig, color}
\usepackage{amsmath}
\usepackage{amsthm}
\usepackage{amssymb}
\usepackage{multirow}
\usepackage{booktabs}
\newtheorem{theorem}{Theorem}

\usepackage{mathtools}
\usepackage{float}
\begin{document}

\title{Massive MIMO Channel Estimation with an Untrained Deep Neural Network}
\author{{
    Eren Balevi, Akash Doshi, and
    Jeffrey G. Andrews}\\	
\thanks{The authors are with the University of Texas at Austin, TX, USA. Email: erenbalevi@utexas.edu, akashsdoshi@utexas.edu, jandrews@ece.utexas.edu. This work has been supported in part by Intel.}				
}

\maketitle 
\normalsize
\begin{abstract}
This paper proposes a deep learning-based channel estimation method for multi-cell interference-limited massive MIMO systems, in which base stations equipped with a large number of antennas serve multiple single-antenna users. The proposed estimator employs a specially designed deep neural network (DNN) to first denoise the received signal, followed by a conventional least-squares (LS) estimation. We analytically prove that our LS-type deep channel estimator can approach minimum mean square error (MMSE) estimator performance for high-dimensional signals, while avoiding MMSE's requirement for complex channel inversions and knowledge of the channel covariance matrix. This analytical result, while asymptotic, is observed in simulations to be operational for just 64 antennas and 64 subcarriers per OFDM symbol. The proposed method also does not require any training and utilizes several orders of magnitude fewer parameters than conventional DNNs.  The proposed deep channel estimator is also robust to pilot contamination and can even completely eliminate it under certain conditions.
\end{abstract}


\section{Introduction}
In multi-antenna systems, obtaining accurate channel state information (CSI) is a central activity both for precoding the spatial streams before transmission and for coherently combining the received signals from each antenna. This is particularly true for massive multi-input multi-output (MIMO) base stations, which are by definition equipped with a very large number of antennas that transmit to many users at the same time and on the same frequency band \cite{Marzetta10}. Channel estimation is nevertheless quite challenging for multicell massive MIMO cellular networks. This is fundamentally due to pilot contamination -- which is the interference of pilot symbols utilized by the users in neighboring cells -- and noise, but also because operations such as matrix inversion and singular value decomposition (SVD) are impractically complex for large channel matrices.  A low overhead, low complexity, and scalable (in terms of the number of antennas) channel estimator is very desirable for massive MIMO and current solutions have nontrivial drawbacks. This paper leverages recent developments in deep learning to design a novel deep massive MIMO channel estimator that achieves these desirable properties. 

\subsection{Related Work and Motivation}
Conventional DNNs are fairly complex and typically require a large number of parameters to be trained with large datasets \cite{Goodfellow-et-al-2016}. Thus, they are not suitable for channel estimation in wireless systems, where channels change quite rapidly.  A recent special DNN design called a \emph{deep image prior} \cite{Ulyanov18} does not require training, and thus avoids the need for a training dataset.  It was proposed to solve inverse problems in image processing such as denoising and inpainting, and is analogous to reducing noise and pilot contamination, which are two key impediments in the channel estimation process. We modify and optimize this architecture for massive MIMO channel estimation so as to have a moderate number of parameters. One of the salient features of our deep channel estimator lies in not requiring any statistical knowledge about the channel except what can be directly obtained from the received signal. This not only eliminates the need to know or learn the channel statistics, but also makes the estimator applicable to any kind of channel including Gaussian or non-Gaussian, line-of-sight (LOS) or non-LOS (NLOS), and limited or rich scattering.

The seminal paper on massive MIMO uses a least-squares (LS) estimator \cite{Marzetta10}. Despite its low complexity, the LS estimator achieves significantly less accurate channel estimation than minimum mean square error (MMSE) estimation  \cite{Beek95}, which has been used in subsequent massive MIMO studies \cite{Ngo13a}, \cite{Ngo13b}, \cite{Adhikary17}. Although the impact of the channel estimation quality is profound in massive MIMO \cite{Ngo13b}, employing an MMSE estimator is undesirable for two main reasons: (i) it requires an accurate estimate of the channel correlation matrix between the base station and each user, the estimation of which requires a very large number of samples in proportion to the number of antennas and has to be repeated frequently due to mobility;  and (ii) a large matrix inversion is needed for MMSE estimation, and thus the complexity growing as the cube of number of antennas.  For both reasons, MMSE estimation scales very poorly in terms of the base station array size.

A key challenge for massive MIMO is pilot contamination, which is a fundamental limiting factor, since small scaling fading and noise vanish as the law of large numbers kicks in \cite{Marzetta10}.  There are many papers that attempt reliable channel estimation for massive MIMO under pilot contamination. The key idea is usually to exploit the differences among the channel covariance matrices of different users. Specifically, \cite{AdhikaryCaire13} partitions users into groups according to the similarity of their covariance matrices, and serves them accordingly. A similar idea was utilized in \cite{YinGesbert13}, which developed a covariance-aware pilot assignment algorithm with some coordination among base stations. A special pilot scheduling algorithm was developed for sparse massive MIMO channels in \cite{YouSwindlehurst16}. The sparsity of massive MIMO channels was also used for channel estimation with complex iterative approximate message passing and expected-maximization (EM) algorithms in \cite{Wen15}. Another method based on channel statistics was presented in \cite{YinHe16}. In addition to these, there are blind channel estimators relying on channel second-order statistics to reduce the number of pilots \cite{Ngo12}, \cite{Muller14}. These works require estimating large covariance matrices or assume they are somehow available for free. Furthermore, their applicability is limited to NLOS zero-mean Gaussian channels, some of which further need sparsity, which exists only for low angle delay spread. 

Developing an improved LS-type estimator, which like LS does not require knowing the channel correlation/covariance matrices and does not involve matrix inversions, is of significant interest but is an open problem.  However, the performance gap is large: we show that the average spectral efficiency decreases by about 50\% for normal LS estimation versus MMSE estimation.  The introduction of techniques from deep learning points to a potential remedy, since these techniques have been recently used for other challenging communication theory problems without closed-form solutions \cite{Oshea17}, \cite{DornerBrink18}, \cite{BaleviGitlin17}, \cite{YeLi18}, \cite{BalAnd19}, \cite{BalAndAntenna19}.  

\subsection{Contributions}
The main contribution of this paper is to propose a novel low complexity massive MIMO channel estimation technique that is robust to pilot contamination. The novelty is the use of a deep neural network (DNN) for denoising prior to a conventional LS-type operation, which is trivially simple.  The proposed  denoising is done via a specially designed DNN similarly to the deep image prior proposed recently for image processing applications \cite{Ulyanov18},  \cite{Heckel18}.  We optimize the number of parameters and epochs to ensure low complexity, eventually reducing the number of parameters from the order of millions to hundreds or a few thousand.  

We mathematically prove that this proposed deep channel estimator approaches and ultimately achieves the MMSE performance as the product of the number of base station antennas, subcarriers and coherence time interval (in terms of OFDM symbols) becomes large. The simulation results appear to confirm this for moderate dimensionality, namely a $64\times64\times64$ signal block, i.e. the number of antennas, subcarriers, and OFDM symbols are all $64$.  Pilot contamination is reduced in the proposed estimator by learning some prior from the interference-free region in the OFDM grid and patching these priors into the pilot contaminated areas. Additionally, we do not assume that the base stations are perfectly synchronized, so the base stations spread pilots randomly and allocate them to users orthogonally over the time-frequency grid for one coherence time interval. Our results reveal that under some conditions (e.g., when 5\% of the OFDM grid is contaminated by neighboring cells with 4 fold weaker interference power relative to the target signal in the low noise regime) the deep channel estimator can completely remove the interference even if the eigenspace of the desired user and interferer fully overlap. The initial results for the proposed deep channel estimator are presented in \cite{BalAndDCE} basically for single antenna OFDM communication without a theoretical analysis and optimizing the architecture. Additionally, \cite{BalAndDCE} does not consider co-channel interference.

The paper is organized as follows. The system model and problem statement are given in Section \ref{System Model} and Section \ref{Problem Statement}. The deep channel estimator is explained in detail in Section \ref{Deep Channel Estimator Model}, an analysis of which appears in Section \ref{Analysis}. The performance results are illustrated with extensive simulations in Section \ref{Simulations}. The paper concludes with Section \ref{Conclusions}.  


\section{System Model} \label{System Model}
We consider a cellular network that has base stations with large number of antennas and single antenna users. Specifically, base stations comprise $M$ antennas and serve $K$ users such that $K \ll M$. We assume that OFDM symbols with $N_f$ subcarriers are transmitted in a time division duplex (TDD) frame structure. To estimate the reciprocal uplink and downlink channels, users in the same cell send orthogonal pilot sequences with length $N_p$. For the target base station the received signal in the frequency domain $\mathbf{Y} \in \mathbb{C}^{M\times N_fN_p}$ can be expressed as
\begin{equation}\label{OFDM_freq}
\mathbf{Y} = \sum_{k=1}^K\sqrt{\rho_{k}}\mathbf{H_{k}}\otimes \mathbf{x_{k}^H} + \sum_{i\in S_{k}}\sqrt{\rho_{i}}\mathbf{H_{i}}\otimes\mathbf{x_{i}^H} + \mathbf{Z}
\end{equation}
where $\rho_{k}$ is the transmit power, $\mathbf{H_{k}} \in \mathbb{C}^{M\times N_f}$ is the channel between the target base station and its $k^{th}$ user, $\mathbf{x_{k}}\in \mathbb{C}^{N_p\times 1}$ is the pilot sequence used for channel estimation such that $\mathbf{x_{k}^H}\mathbf{x_{k}} = N_p$ and $\otimes$ denotes the Kronecker product. The notation is the same for the second term in the right-hand side (RHS) of \eqref{OFDM_freq}, which represents the users in other cells, and 
\begin{equation}\label{set}
 S_{k} = \{ i | \mathbf{x_{i}} = \mathbf{x_{k}}, i\neq k \}.
\end{equation}
The last term $\mathbf{Z} \in \mathbb{C}^{M\times N_fN_p}$ denotes the Gaussian noise matrix whose independent and identically distributed (i.i.d.) elements are zero-mean Gaussian random variables with variance $\sigma^2$. 

The $k^{th}$ user signal in the base station is obtained by 
\begin{equation}\label{kthsignal}
\mathbf{Y_{k}} = \mathbf{Y}(\mathbf{I_{N_f}}\otimes\mathbf{x_{k}} )
\end{equation}
such that $\mathbf{Y_{k}} \in \mathbb{C}^{M\times N_f}$. Due to the mixed-product property of the Kronecker product
\begin{eqnarray}
( \mathbf{H_{k}} \otimes \mathbf{x_{k}^H} )(\mathbf{I_{N_f}}\otimes\mathbf{x_{k}} ) &=& (\mathbf{H_{k}} \mathbf{I_{N_f}}) \otimes (\mathbf{x_{k}^H} \mathbf{x_{k}} ) \\ \nonumber
& = &  N_p\mathbf{H_{k}},
\end{eqnarray}
it is straightforward to express \eqref{kthsignal} as
\begin{equation}\label{kthsignalopen}
\mathbf{Y_{k}} = \sqrt{\rho_{k}}N_p\mathbf{H_{k}} + \sum_{ i \in S_k }\sqrt{\rho_{i}}N_p\mathbf{H_{i}} + \mathbf{Z_{k}}
\end{equation}
where
\begin{equation}
\mathbf{Z_{k}} = \mathbf{Z}(\mathbf{I_{N_f}}\otimes\mathbf{x_{k}} ).
\end{equation}
As can be observed in \eqref{set}, other users in other base stations can also use the same pilot sequences with the $k^{th}$ user in the target cell. This is because pilots are limited by the time-frequency resources, and so it is not possible to allocate orthogonal pilots for all users in all cells at least not without greatly degrading the ability to transmit information-bearing symbols. The resulting interference is known as pilot contamination. 

\section{Problem Statement} \label{Problem Statement}
To have more compact expressions, the matrices are defined as vectors by concatenating the columns, which are given by
\begin{equation}
\mathbf{\underline{Y}_{k}} = \text{vec}( \mathbf{Y_{k}} )
\end{equation}
where $\mathbf{\underline{Y}_{k}} \in \mathbb{C}^{MN_f}$. The same notation is utilized for $\mathbf{\underline{H}_{k}}$, $\mathbf{\underline{H}_{i}}$ and $\mathbf{\underline{Z}_{k}}$. Substituting \eqref{kthsignalopen} with these yields
\begin{equation}\label{kthsignalvec}
\mathbf{\underline{Y}_{k}} = \sqrt{\rho_{k}}N_p\mathbf{\underline{H}_{k}} + \sum_{ i \in S_k }\sqrt{\rho_{i}}N_p\mathbf{\underline{H}_{i}} + \mathbf{\underline{Z}_{k}}.
\end{equation}
To estimate the channel between the $k^{th}$ user and the target base station, \eqref{kthsignalvec} is multiplied with a linear matrix such that
\begin{equation}
\mathbf{\underline{\hat{H}}_{k}} = \mathbf{A_{k}}\mathbf{\underline{Y}_{k}}
\end{equation}
where
\begin{equation}
\mathbf{\underline{\hat{H}}_{k}} = \text{vec}(\mathbf{\hat{H}_{k}} )
\end{equation}
and
\begin{align} \label{channel_estimators}
    \mathbf{A_{k}} = 
    \begin{cases}
			 \frac{1}{\sqrt{\rho_{k}}N_p}\mathbf{I_{MN_f}}, \text{for LS estimation} \\
       \sqrt{\rho_{k}}\mathbf{R_{\underline{H}_{k}}}( \Gamma_{k} + \sigma^2\mathbf{I_{MN_f}} )^{-1}, \text{for MMSE estimation} 
    \end{cases} 
\end{align}
in which
\begin{equation}
\Gamma_{k} = \sum_{  i \in S_k } \sqrt{\rho_{i}}N_p\mathbf{R_{\underline{H}_{i}}}.
\end{equation}
As is clear from \eqref{channel_estimators}, LS estimation has very low complexity, whereas MMSE estimation requires not only the autocorrelation matrices of all users that use the same pilot sequence but also a matrix inversion, the complexity of which scales as $(MN_f)^2$. Hence, the MMSE estimator is not a viable option for systems with large number of antennas $M$ and/or subcarriers $N_f$ \cite{Wen15}. Despite the appeal of the LS estimator in terms of low complexity, it provides much less accurate estimation. To illustrate this, we consider average spectral efficiency
\begin{equation}\label{SE}
\eta = \frac{N_p}{N}\mathbb{E}[\log_2(1+\text{SINR})]
\end{equation}
where $N$ is the coherence time interval. The average sum of the spectral efficiency based on \eqref{SE} for LS and MMSE estimators is depicted for different combiners, namely for maximum ratio (MR), zero-forcing (ZF), and MMSE combiners in Fig. \ref{fig:MMSE_vs_LS}. As can be shown, there is a considerable decrease in the average sum spectral efficiency due to LS channel estimation, in particular for MMSE and ZF combiners. A channel estimation technique that exhibits MMSE estimator performance with LS estimator complexity is highly desirable.
\begin{figure} [!t] 
\centering 
\includegraphics [width=5in]{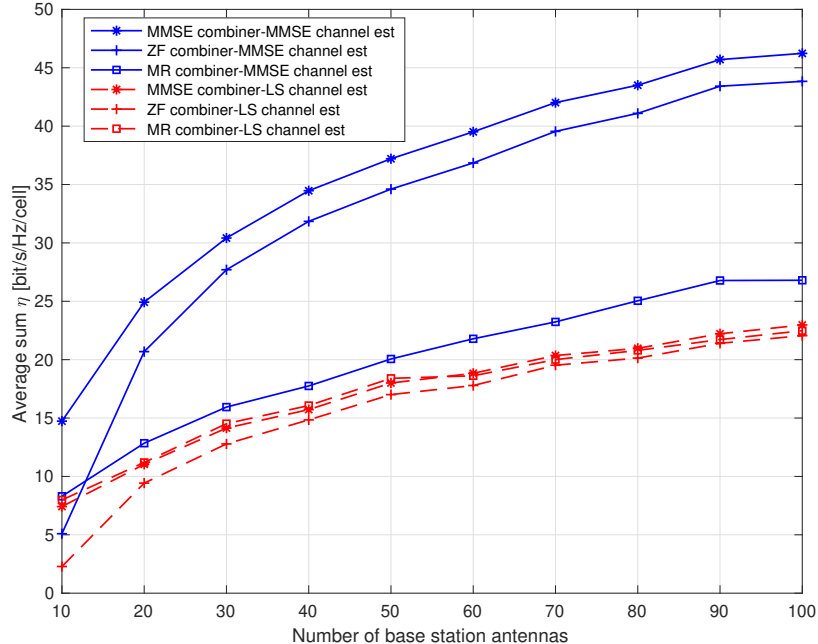}
\caption{Average sum spectral efficiency of LS and MMSE channel estimators for different combining techniques with increasing number of antennas.}\label{fig:MMSE_vs_LS}
\end{figure}

We consider deep learning as a remedy, however the high dimensionality of the signals is a challenge. This is because the higher the signal dimension is, the larger the number of necessary parameters in the DNN model, which needs to be trained with a dataset whose size is proportional to the number of parameters. To illustrate, a fully connected neural network for an $M$ antenna OFDM system requires $U=MN_f$ input neurons. If there are $l$ layers in this DNN, each of which has $k_iU$ units for $i=0, 1, \cdots, l-1$, this leads to $\sum_{i=0}^{l-1} k_ik_{i+1}U^2$ parameters, where $k_0=k_l=2$ due to the real and imaginary parts of the signal. This can easily yield millions of parameters, and thus requires a very large training dataset. To illustrate, if $M=64$ and $N_f=1024$, this yields approximately $5*10^{10}$ parameters for $6$ layers when $k_i = 2$ for $i=0, 1, \cdots 6$. Although convolutional neural networks can considerably decrease the number of parameters, a large training dataset is still necessary. This is obviously an impediment in using neural networks for real-time channel estimation, where only a very limited number of pilots (i.e. labels)\footnote{There can be some unsupervised or semi-supervised learning models that make channel estimation with no labels or with very limited labels. However, there is not any generic known channel estimation model yet for this method, and this subject remains mostly open.} can be used. 

In this paper, we propose a new DNN based channel estimation method \textbf{that does not require training}. Our main idea is to denoise the received signal via the DNN and then use that denoised signal for LS channel estimation instead of the raw received signal. Since the proposed estimator does not require training, there is no complexity increase due to training. This also prevents the inevitable performance loss for estimators that are trained for some channel realizations but then used in others. The details of the proposed method are elaborated next.

\section{Deep Channel Estimator Model} \label{Deep Channel Estimator Model} 
Training overhead is the primary obstacle to making state-of-the-art DNNs practically implementable for high-dimensional channel estimation. In the context of image processing, a recent paper shows that training is not necessary for a special DNN design, which is known as Deep Image Prior (DIP), for solving the inverse problems of denoising and inpainting \cite{Ulyanov18}. The main idea behind this untrained DNN or DIP model is to fit the parameters of a neural network for each image on the fly without training them on large datasets beforehand. This model was later optimized to reduce the number of required parameters \cite{Heckel18}. Both \cite{Ulyanov18} and \cite{Heckel18} observed very efficient denoising and inpainting performance thanks to the specifically designed DNN architecture, which has low impedance for natural images and high impedance against noise. 

For massive MIMO-OFDM channel estimation, denoising and inpainting are analogous to eliminating noise and pilot contamination, respectively, and adapting DIP model to the channel estimation problem is promising. This is because (i) in communication systems, there is a limited number of pilots (or labeled data), and thus the architectures based upon large training dataset are not feasible; (ii) in conventional DNNs, training and testing have to be done for the same channel realization to obtain better performance, which brings in heavy training overhead; and (iii) noise and interference are the main impediments that hinders to make a reliable channel estimation for massive MIMO. Motivated by these factors, the specifically designed DNN architecture for the DIP model is leveraged to make channel estimation. In particular, we modify the input and output layers of the one variant of DIP architecture \cite{Heckel18}, and use it as a baseline, which we term a \textit{deep channel estimator}.

The proposed deep channel estimator is composed of two stages. In the first stage, a less noisy signal is generated from the received signal through a specially designed DNN architecture mentioned above, and some prior information is obtained to mitigate interference. In the second stage the generated or filtered signal is multiplied by the Hermitian of the pilot sequence for channel estimation. This apparently means that we propose an LS-type channel estimator with the only difference being that the signal generated by the DNN is used instead of the received signal. By doing that the low complexity nature of LS estimator is combined with the noise reduction capability of the DNN so as to have a near MMSE estimation performance. The price paid for the proposed deep channel estimator is the need for fitting the parameters of the DNN periodically for each OFDM grid, whose period is determined by the channel coherence time (or equivalently maximum Doppler spread). However, the complexity increase is quite reasonable thanks to the low number of parameters, as will be explained. 

The received signal in \eqref{OFDM_freq} can be equivalently written in 3-dimensional form as 
\begin{equation} \label{OFDM_grid}
\mathbf{Y} = \{\{\{Y[m,q,n]\}_{m=1}^{M}\}_{q=1}^{N_{f}}\}_{n=1}^{N}.
\end{equation}
where $Y[m,q,n]$ is the received signal of the target base station in the $m^{th}$ antenna for the $q^{th}$ subcarrier in the $n^{th}$ OFDM symbol. Notice that \eqref{OFDM_grid} is expressed in terms of the length of the coherence time instead of the number of pilots, which contains $N$ OFDM symbols. This is because the parameters of the DNN has to be fitted periodically with coherence time. The real and imaginary part of \eqref{OFDM_grid} is separated into $2$ independent channels in our architecture, since tensors do not support complex operations. This tensor representation of $\mathbf{Y}$ is denoted as $\mathbf{Y_T}$. Specifically, $\mathbf{Y_T} \in \mathbb{R}^{M\times N_f\times N \times 2}$, where the dimensions are for the spatial, frequency, time, and complex domains. In our architecture, we stack the spatial and complex domains which leads to $\mathbf{Y_T} \in \mathbb{R}^{N_s\times N_f\times N }$, where $N_s = M \times 2$.

The working principle of the deep channel estimator is to generate $\mathbf{Y_T}$ from a randomly chosen input tensor $Z_0$, which can be considered as an input filled with uniform noise, through hidden layers, whose weights are also randomly initialized, and then optimized via gradient descent. The overall DNN model that depicts the input, output and hidden layers for a $3$-dimensional communication signal is given in Fig. \ref{fig:e2emodel}.
\begin{figure*} [!t] 
\centering 
\includegraphics [width=5in]{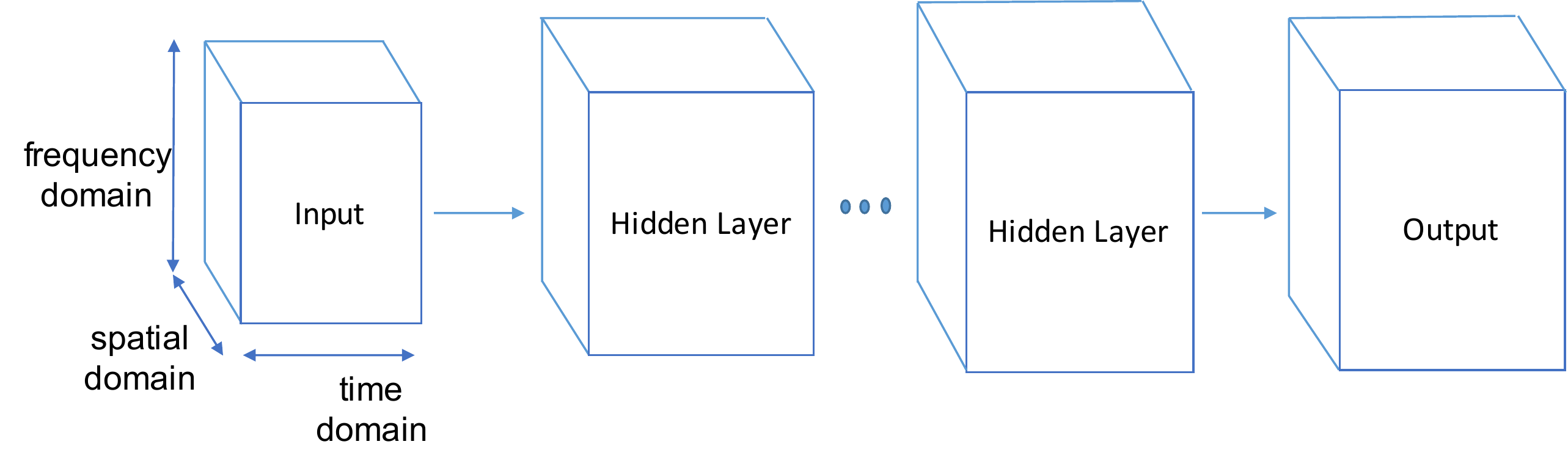}
\caption{The DNN architecture to denoise and inpaint the received signal before channel estimation for a $3$-dimensional communication signal.} \label{fig:e2emodel}
\end{figure*}
The key component in the aforementioned DNN model is the hidden layers, which are composed of four major components. These are: (i) a $1\times1$  convolution, (ii) an upsampler, (iii) a rectified linear unit (ReLU) activation function, and (iv) a batch normalization. A $1\times1$ convolution means that each element in the time-frequency grid is processed with the same parameters through the spatial domain, which changes the dimension. More precisely, an $N_s^{(i-1)}\times1\times 1$ data vectors in the $i^{th}$ hidden layer is element-wise multiplied with an $N_s^{(i-1)}\times1\times 1$ kernel and summed. There are $N_s^{(i)}$ different kernels, which are shared for each slot in the time-frequency axes. Hence, the spatial dimension becomes $N_s^{(i)}$. This can be equivalently considered as each vector in the time-frequency slot being multiplied with the same (shared) $N_s^{(i)} \times N_s^{(i-1)}$ matrix. In what follows, upsampling is performed to exploit the couplings among neighboring elements in the time and frequency grid. More precisely, the time-frequency signal is upsampled with a factor of $2$ via a bilinear transformation. Next, the ReLU activation function is used to make the model more expressive for nonlinearities. The last component of a hidden layer makes batch normalization for a batch size of $1$ to avoid vanishing gradients. This structure of a hidden layer is portrayed in Fig. \ref{fig:model}. All the hidden layers have the same structure except for the last hidden layer, which does not have an upsampler.
\begin{figure*} [!t] 
\centering 
\includegraphics [width=6in]{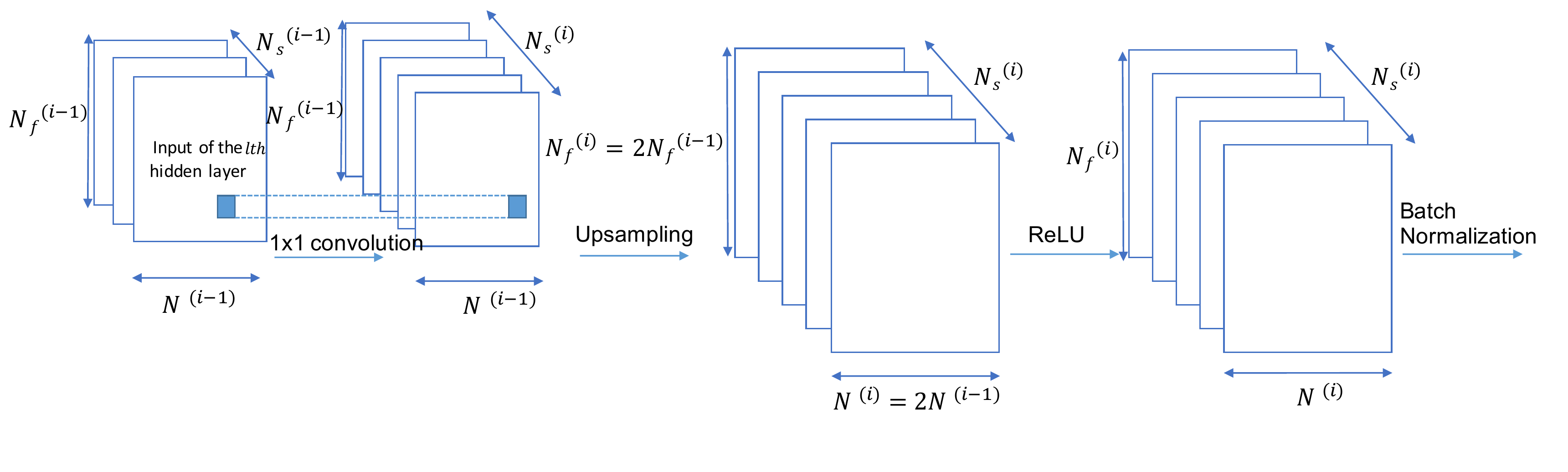}
\caption{The structure of the $i^{th}$ hidden layer, whose input dimension is $N_s^{(i-1)} \times N_f^{(i-1)} \times N^{(i-1)}$ and output dimension is $N_s^{(i)} \times N_f^{(i)} \times N^{(i)}$. Note that $N_f^{(i)} = 2N_f^{(i-1)}$ and $N^{(i)} = 2N^{(i-1)}$. The spatial dimensions $N_s^{(i-1)}$ and $N_s^{(i)}$ are the hyperparameters that are used by the $1\times1$ convolution operations.} \label{fig:model}
\end{figure*}

The mathematical representation of the aforementioned architecture is given next. Accordingly, the tensor $\mathbf{Y_T}$ is parameterized for the $l+1$ layer as
\begin{equation}
\mathbf{\hat{Y}_T} = f_{\theta_l}(f_{\theta_{l-1}}( \cdots f_{\theta_0}(Z_0)))
\end{equation}
where the input $Z_0$ has a dimension of $N_s^{(0)} \times N_f^{(0)} \times N^{(0)}$ in the spatial, frequency and time domain, respectively. These dimensions are determined according to the number of hidden layers and the output dimension, in which $N_s^{(l)}=N_s$,  $N_f^{(l)}=N_f$, $N^{(l)}=N$. The layers from $0$ to $l-1$ are counted as a hidden layer, and for $i=0,1,\cdots,l-2$
\begin{equation} \label{eq:Op1}
f_{\theta_{i}}=\text{BatchNorm}(\text{ReLU}(\text{Upsampler}(\theta_i\circledast Z_{i}))
\end{equation}
where $Z_i$ is the input of the $i^{th}$ hidden layer, $\theta_i$ are the parameters, and $\circledast$ represents the so-called ``convolution'' operator, which actually refers to cross-correlation in signal processing. More precisely, a $1\times1$ convolution is utilized as a cross-correlator, which means that the spatial vector for each element of the time-frequency grid is multiplied with the same shared parameter matrix to obtain the new spatial vector for the next hidden layer. The last hidden layer is
\begin{equation} \label{eq:Op2}
f_{\theta_{l-1}}=\text{BatchNorm}(\text{ReLU}(\theta_{l-1}\circledast Z_{l-1}),
\end{equation}
and the output layer is
\begin{equation}
f_{\theta_{l}}=\theta_{l}\circledast Z_{l}.
\end{equation}
All the parameters can be represented as 
\begin{equation}
\Theta = (\theta_{0}, \theta_{1}, \cdots, \theta_{l}),
\end{equation}
which are optimized according to the square of $l_2$-norm
\begin{equation}\label{obj_func}
\Theta^* = \arg \min_\Theta ||\mathbf{Y_T}-\mathbf{\hat{Y}_T}||_2^2.
\end{equation}
The output of the DNN for the optimized parameters is
\begin{equation}\label{outDNN}
\mathbf{Y_T^*} = f_{\Theta^*}(Z_0)
\end{equation}
where $ \mathbf{Y_T^*}= \{\{\{Y^*[m, q, n]\}_{m=1}^{M}\}_{q=1}^{N_{f}}\}_{n=1}^{N}$. After generating the denoised signal in \eqref{outDNN} from a random input $Z_0$, an LS channel estimator is employed by multiplying \eqref{outDNN} with the Hermitian transpose of the pilot sequence. 

\section{Theoretical Analysis} \label{Analysis} 
The denoising capability of the proposed LS-type deep channel estimator determines how close it can approach the MMSE estimation performance. Next, we prove that our architecture can filter all the noise for high-dimensional signals, e.g., for massive MIMO-OFDM, and can achieve the MMSE estimator performance.

\begin{theorem}
The proposed LS-type deep channel estimator achieves the MMSE estimator performance as the product of the number of base station antennas $M$, number of subcarriers $N_f$ and coherence time interval $N$ goes to infinity, assuming there is no pilot contamination. That is,
\begin{equation} \label{dce_mmse}
\lim_{MN_fN \rightarrow \infty}\epsilon_{\rm dce} = \lim_{MN_fN \rightarrow \infty}\epsilon_{\rm mmse} = 0
\end{equation}
where $\epsilon_{\rm dce}$ and $\epsilon_{\rm mmse}$ are the channel estimation errors for the proposed deep channel estimator and conventional MMSE estimator, respectively.
\end{theorem}

\begin{proof}
The proof is composed of three parts. In the first part, we generalize the noise suppression level of the architecture \cite{Heckel18} for the deep channel estimator as 
\begin{equation} \label{noise_sup_inequality}
\underset{ \Theta }{\mathrm{min\ }} \frac{||\textbf{n} - \textbf{n}_{\rm fit}(\Theta) ||^2}{||\textbf{n}||^2} \geq  1 - c\left(\frac{{\left( \prod_{i=0}^{l-1}N_s^{(i)}\right)^{\frac{2}{l}} \log\left(\prod_{i=0}^{l-1}N_f^{(i)}N^{(i)}\right)}}{N_s^{(l)}N_f^{(l)}N^{(l)}}\right),
\end{equation}
where $\textbf{n}$ is the noise vector in the received signal, and $\textbf{n}_{\rm fit}(\Theta)$ shows the fitted amount of noise at the output of the deep channel estimator such that $||\textbf{n}_{\rm fit}(\Theta)||=0$ means that all noise is cancelled, and $c$ is a numerical constant. Although \eqref{noise_sup_inequality} is satisfied with probability at least $1-2(N_f^{(0)}N^{(0)})^{-\left( \prod_{i=0}^{l}N_s^{(i)}\right)^{\frac{2}{l+1}}}$ for $N_s^{(l)}=1$ in \cite{Heckel18}, this probability goes to 1 for our case due to the high-dimensional massive MIMO-OFDM signal model, i.e.,
\begin{equation}
\lim_{M \rightarrow \infty} \left\{1-2(N_f^{(0)}N^{(0)})^{-\left( \prod_{i=0}^{l}N_s^{(i)}\right)^{\frac{2}{l+1}}} \right\} = 1
\end{equation}
since $M=N_s^{(l)}/2$.
Hence, with the optimum parameters $\Theta^*$, \eqref{noise_sup_inequality} can be expressed in terms of the maximum noise suppression level
\begin{equation} \label{noise_sup}
n_{\rm{supp}} = \inf\left\{\frac{||\textbf{n} - \textbf{n}_{\rm fit}(\Theta^*) ||^2}{||\textbf{n}||^2}\right\} = 1 - c\left(\frac{{\left( \prod_{i=0}^{l-1}N_s^{(i)}\right)^{\frac{2}{l}} \log\left(\prod_{i=0}^{l-1}N_f^{(i)}N^{(i)}\right)}}{N_s^{(l)}N_f^{(l)}N^{(l)}}\right)
\end{equation}
with probability 1. 

In the second part of the proof, we make use of deep learning theory regarding overparameterization. We observe that the denominator of the second term in the right-hand side of \eqref{noise_sup} scales with the dimension of the received signal, since $N_s^{l}=2M$, $N_f^{l}=N_f$, $N^{l}=N$, whereas the dimension of the hidden layers does not. In particular, 
\begin{equation} \label{time-freq-scale}
\log\left(\prod_{i=0}^{l-1}N_f^{(i)}N^{(i)}\right) = \log\left( \frac{N_f^{(l-1)}N^{(l-1)}}{\prod_{l=0}^{l-1}4^l}\right)^l 
\end{equation}
due to \eqref{eq:Op1} and \eqref{eq:Op2}, in which the denominator of the right-hand side of \eqref{time-freq-scale} is scaled by 4 due to the oversampling by 2 in the time and frequency axes. Since $N_f^{(l)}N^{(l)}\gg l$,
\begin{equation} \label{time_freq_inc}
 \log\left( \frac{N_f^{(l-1)}N^{(l-1)}}{\prod_{l=0}^{l-1}4^l}\right)^l  \ll N_f^{(l)}N^{(l)}.
\end{equation}
Now, we proceed to see how the spatial dimension of the hidden layers scales with the number of antennas. Accordingly, the objective function in \eqref{obj_func} is written in terms of energy minimization \cite{Ulyanov18} 
\begin{equation}\label{energy_min}
\min_\Theta E(\mathbf{Y_T}, \mathbf{\hat{Y}_T}).
\end{equation}
It is standard to express \eqref{energy_min} in terms of a function approximator $G(\cdot)$ and a regularizer $R(\cdot)$ as
\begin{equation}\label{DNN_eqv}
\min_\Theta E(\mathbf{Y_T},  G({\Theta})) + R(G({\Theta})).
\end{equation}
For \eqref{DNN_eqv}, increasing the width of the last hidden layer while keeping the dimension of the other hidden layers fixed is sufficient to fit the received signal \cite{GlorotBengio10}. Using \eqref{DNN_eqv} by defining
\begin{equation}
R(G({\Theta})) = ||\Theta||_F^2 + <\mathbf{A}, \Theta^T\Theta>
\end{equation}
where $\mathbf{A}$ is a random positive semidefinite matrix with arbitrarily small Frobenius norm, and writing the last layer of the deep channel estimator for a time-frequency slot as 
\begin{equation}
(N_s^{(l)})^2=(\theta_{l}\circledast N_s^{(l-1)})^2,
\end{equation}
it is sufficient to increase the spatial dimension as \cite{DuLee18}
\begin{equation} \label{spatial_inc}
N_s^{(l-1)} = \sqrt{2r}
\end{equation}
where $r$ is the rank of the channel that shows the number of independent received samples. That is, it increases sublinearly with the increasing number of antennas. 
Due to \eqref{time-freq-scale} and \eqref{spatial_inc}, 
\begin{equation}\label{noise_sup_1}
\lim_{MN_fN \rightarrow \infty}n_{\rm{supp}} = 1.
\end{equation}

In the third part, we derive the asymptotic channel estimation errors in view of the first two parts. The channel estimation error for MMSE estimator can be expressed in terms of covariance matrix  
\begin{equation}\label{mmse}
\epsilon_{\rm{mmse}} = \rm{tr}(\mathbf{C})
\end{equation}
where
\begin{equation} \label{Cmmse}
\mathbf{C} = \mathbf{R}-\mathbf{R}\left(\mathbf{R}+\frac{1}{\rm SNR}\mathbf{I}\right)^{-1}\mathbf{R}.
\end{equation}
In terms of the eigenvalues of the correlation matrix $\mathbf{R}$, \eqref{mmse} can be written using \eqref{Cmmse} as
\begin{equation}
\epsilon_{\rm{mmse}} = \sum_{i=1}^{\rm{rank}(\mathbf{R})}\lambda_m-\frac{\lambda_m^2}{\lambda_m+1/{\rm SNR}}.
\end{equation}
Thus, 
\begin{equation}
\lim_{MN_fN \rightarrow \infty}\epsilon_{\rm{mmse}} = 0,
\end{equation}
since uncorrelated noise vanishes in massive MIMO. In the case of LS estimator, the error is equal to 
\begin{equation}\label{ls}
\epsilon_{\rm dce} = \epsilon_{\rm{ls}} =  \frac{\rm{rank}(\mathbf{R})}{\rm SNR},
\end{equation}
and 
\begin{equation}
\lim_{MN_fN \rightarrow \infty}\epsilon_{\rm{dce}} = 0,
\end{equation}
due to \eqref{noise_sup_1}. This completes the proof of \eqref{dce_mmse}. 
\end{proof}

Notice that if the spatial dimensions of all hidden layers are increased equally instead of only increasing the last hidden layer spatial dimension, this would result in less increase than  \eqref{spatial_inc}. This means that the spatial dimension increases at worst with the square root of the rank of the channel\footnote{Our empirical results support this argument. To illustrate, $N_s^{(l-1)} = 8$ for single antenna, whereas $N_s^{(l-1)} = 16$ for 64 antennas.}. Another important point regarding this theorem is that the deep channel estimator can ultimately achieve zero estimation error without increasing the transmission power, instead just by increasing the number of antennas and subcarriers.

Even if there is pilot contamination in the environment, the proposed estimator can inherently resist (and even completely eliminate) interference up to some point. However, this holds only if the interference exists in some limited region of the OFDM grid of the desired signal. We associate this behavior with the inpainting capability of the DIP architecture \cite{Ulyanov18}. This implies that our LS-type estimator under pilot contamination can give the MMSE estimation performance in single-cell massive MIMO even for the multicell case, if the pilot contamination is sufficiently localized in time and frequency. This success can be attributed to learning prior information from some interference-free regions and then patching this prior information into the interference regions. In this sense, it is similar to dictionary learning \cite{NIPS2008_3448}. The comparison of various dictionary learning methods with our estimator as well as integrating our model into one of the dictionary learning methods for enhanced interference mitigation are left as future work.  

\section{Simulations} \label{Simulations}
The proposed deep channel estimator is compared with the traditional LS and MMSE channel estimators given in \eqref{channel_estimators} using the ``LTE-Extended Pedestrian A Model (EPA)'' and ``Kronecker'' channel model. The performance metric is the normalized mean square error (NMSE), which is defined as 
\begin{equation} \label{eq:NMSE}
\text{NMSE} = \mathbb{E}\left[\frac{||\mathbf{\underline{H}_{j,k}}-\mathbf{\underline{\hat{H}}_{j,k}}||_2^2}{||\mathbf{\underline{H}_{j,k}}||_2^2}\right]
\end{equation}
where $\mathbf{\underline{H}_{j,k}}$ and $\mathbf{\underline{\hat{H}}_{j,k}}$ are the column vectors that specify the actual and the estimated channel taps in the frequency domain over all antennas, respectively. In this section, we first state the experimental details, then provide the simulation results and discuss the complexity of the estimator. 

\subsection{Experimental Details}
The deep channel estimator is implemented in Pytorch \cite{paszke2017pytorch} with $6$ hidden layers, i.e., $l=6$ as described in Section \ref{Deep Channel Estimator Model}. Without any loss of generality, the spatial dimension of the hidden layers is taken as $N_s^{(i)}=k$ for $i = 0,1,\cdots,l-1$. Then, the number of parameters (or equivalently the value of $k$) is optimized using two Nvidia GeForce GTX 2070 GPUs for acceleration\footnote{Note that the dimension of the time and frequency axes of the hidden layers are not tunable, since these are determined by the size of received signal matrix and the number of hidden layers.}. The performance of our estimator is evaluated for two channel models, namely the LTE-EPA and Kronecker channel model, which is commonly used to model MIMO channels. However, we present most of our results only for the LTE-EPA model, because our empirical results show that there is not any significant performance difference between these two channel models. To generate a channel realization for the LTE-EPA model, we use the MATLAB\textsuperscript{\textregistered} LTE Toolbox, and obtain an $M \times 64\times64$ (antennas $\times$ subcarriers $\times$ symbols) channel matrix assuming that coherence time interval is larger than or equal to $64$ OFDM symbols. For the Kronecker channel model, we assume an exponential spatial correlation matrix at the base station with correlation coefficient $\rho=0.5$ without any loss of generality \cite{Loyka01}. 

As is the case for multi-cell massive MIMO, users in the same cell are assigned to the orthogonal pilot signals that can be non-orthogonal to the users in the neighboring cell. Our estimator does not put any constraint to the pilot arrangement, since pilots are not used while fitting the parameters of the DNN. Specifically, pilots are only used to perform LS estimation after the received signal is filtered via the DNN. To be more practical in the sense of not requiring any tight synchronization among base stations, a random pilot allocation is utilized for multi-cell massive MIMO such that each base station randomly and orthogonally spreads the pilot tones for its users throughout the OFDM grid per coherence time interval\footnote{Note that a block type pilot arrangement in which the the pilots are sent at the beginning of each coherence time interval gives the same performance with the same amount of pilot tones spread randomly throughout the OFDM grid in one coherence time interval.}. This can be easily done if the coherence bandwidth is equal to one subcarrier; otherwise more careful allocations have to be done so as to have a single pilot per coherence bandwidth. Random pilot arrangement is similar to the tone hopping (see \cite{tse2005fundamentals}), which is proposed to attain ergodic capacity bounds and diminish the impact of ``deep fades".

To simulate pilot contamination, we have assumed that a selected number of resource elements in the OFDM grid has a single dominant interferer outside the cell of the desired user signal. This dominant signal is chosen to have a signal of random QPSK symbols multiplied with its complex channel matrix which is simply another realization generated by the LTE-EPA model. This is indeed the worst-case scenario, because the covariance matrices of the user and interferer are fully overlapped. In the simulations, we also consider another pilot contamination scenario such that there is a contiguous interference both in the time and frequency domain over some number of resource elements in the OFDM grid. 

\subsection{Results} \label{Results}
The performance of the proposed channel estimator is first observed for traditional single antenna communication such that single antenna users transmit/receive OFDM symbols to/from a single antenna base station, i.e., $M = 1$. This enables us to quantify how the number of parameters scales with the number of antennas. We then proceed to the case of single cell massive MIMO. In other words, we consider the hypothetical scenario where all users in the cell are assumed to have orthogonal pilots and there is no pilot contamination at the base station. We finally demonstrate the robustness of our estimator in a multi-cell massive MIMO system, where pilot contamination occurs at the base station due to non-orthogonal pilot sequences employed by users in the neighboring cells. 

\subsubsection{Single antenna OFDM Communication}
We first highlight our results in the case of one antenna on the base station in the presence and absence of co-channel interference. For this case, the received signal or the output of the deep channel estimator is chosen to be a $2\times64\times64$ matrix, where 2 represents the real and imaginary part, the first 64 represents the number of subcarriers, and the next 64 is the number of OFDM symbols. More precisely, our architecture performs the operation outlined in \eqref{eq:Op1} for $l-1=5$ times, then the last hidden layer calculates the expression in \eqref{eq:Op2}, and finally the output layer brings the output signal to the required channel matrix size, i.e., $2\times64\times64$.

To find the optimum number of parameters in the absence of co-channel interference, we simply add AWGN to the desired signal, and adjust its variance so as to have an SNR between 0 and 20 dB range. In order to optimize the number of channels per layer or the value of $k$, we take a single channel realization disturbed by the least noise (i.e the highest SNR in our range) and observe the convergence of its NMSE with the number of epochs by performing Adam optimization \cite{kingma2014adam} with a learning rate of $0.01$. We find the number of epochs at which the lowest NMSE is achieved for a given $k$, and proceed to denoise the received signal for the aforementioned range of SNRs for the calculated number of epochs. This approach is often referred to as \textit{early stopping}. The number of epochs is tabulated in Table \ref{tab:hpt} as a function of $k$ and the total number of weights in the architecture. 

\begin{table}[H]
\caption{Hyperparameter Tuning for $M=1$}
\label{tab:hpt}
\begin{center}
\begin{tabular}{|c|c|c|}
\hline
$k$ & Epochs & Total Weight Count\\
\hline
8&2000&496\\
\hline
16&1300&1760\\
\hline
32&900&6592\\
\hline
64&250&25472\\
\hline
\end{tabular}
\end{center}
\end{table}

\begin{figure*}[!h]
\centering
\subfigure[Without interference]{
\label{fig:nik}
\includegraphics[width=5in]{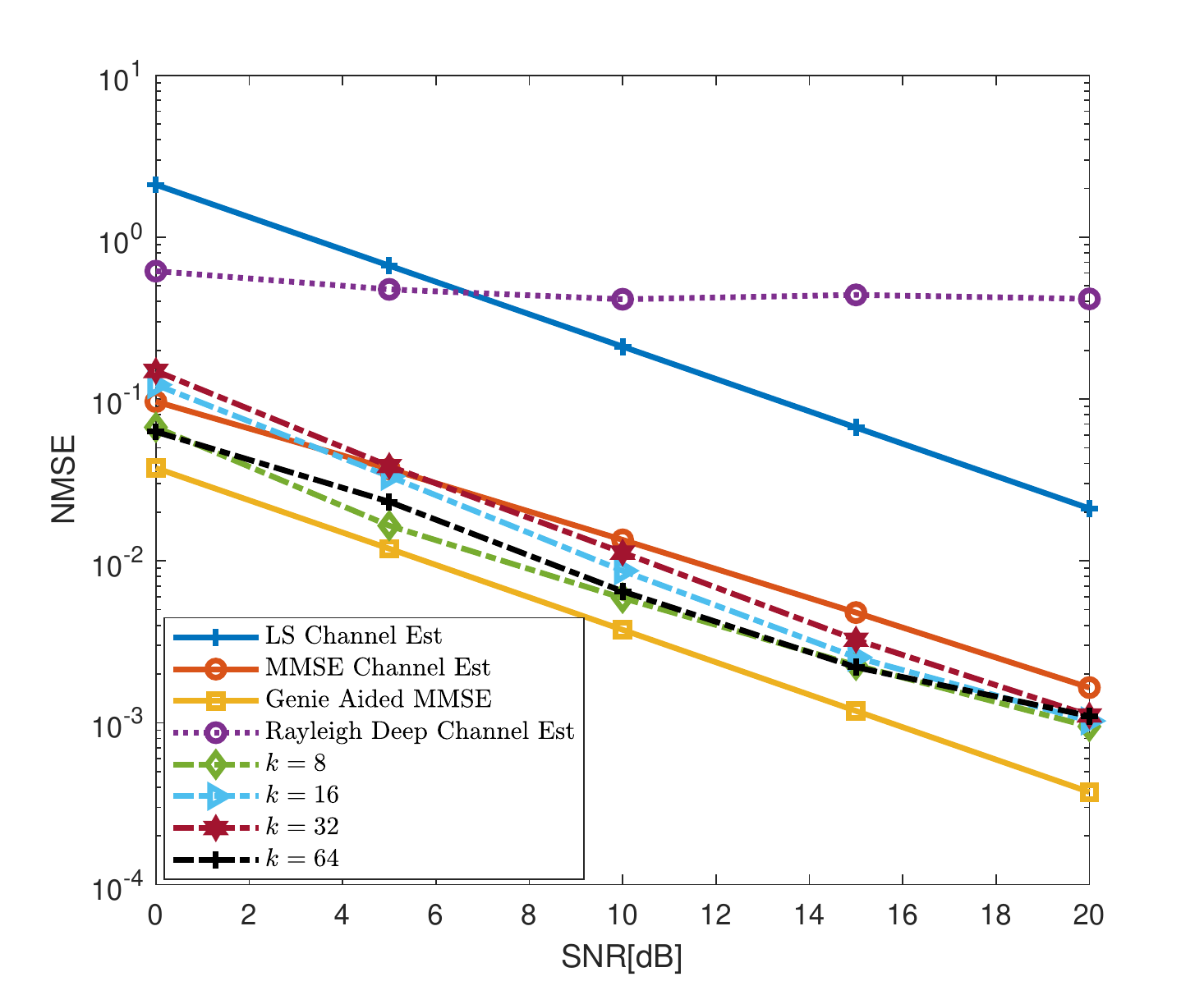}}
\qquad
\subfigure[With co-channel interference (SIR = 6 dB)]{
\label{fig:ik}
\includegraphics[width=5in]{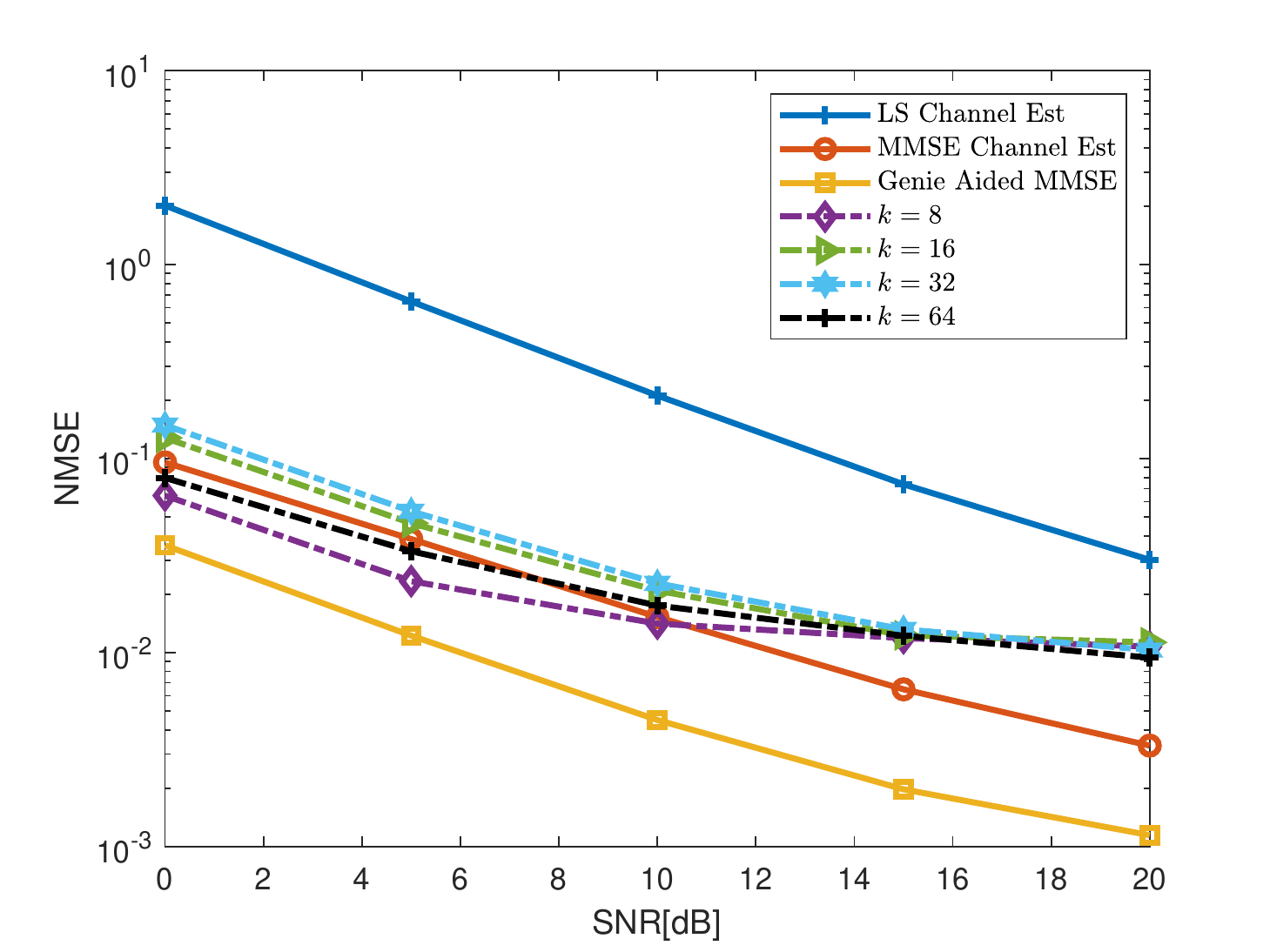}}
\caption{NMSE of the proposed estimator for different $k$ and $M=1$ with respect to SNR in comparison to LS and MMSE estimators.}\label{fig:k}
\end{figure*}
As depicted in Fig. \ref{fig:nik}, the NMSE is lowest for $k = 8$, gets progressively higher for $k = 16$ and $k = 32$, and once again decreases for $k = 64$, almost equal to $k = 8$ at an SNR of 0 dB. However, there is very little to tell apart the different architectures at SNR of 20 dB. These performance statistics can be somewhat explained by the following insight: larger noise levels require smaller values of $k$. If the noise is significantly larger, then we can either choose smaller $k$ or use early stopping. In this plot, the MMSE curve is obtained using the channel correlation matrix that is computed via Monte Carlo simulation as outlined in \cite{Beek95}, whereas the ``Genie Aided MMSE" assumes that the channel correlation matrix is available for free, which is highly impractical. Promisingly, our channel estimator for $k=8$ clearly outperforms LS and MMSE estimators and approach the `Genie Aided MMSE" performance without having any (statistical) information other than the received signal. To have a better understanding of why the deep channel estimator works so well, we observe its performance for an unrealistic case, in which each subcarrier in the OFDM grid has an i.i.d Rayleigh fading channel. In this case, our estimator does not perform well, which proves that its success is attributed to exploiting correlations.

To find the optimum number of parameters in the presence of co-channel interference, 10\% of the resource elements of the OFDM grid expressed in \eqref{OFDM_grid} are corrupted by injecting an interference that is $6$ dB weaker than the desired signal, i.e.,  SIR = 6 dB. As shown in Fig. \ref{fig:ik}, clearly $k=8$ outperforms $k=64$ for SNRs less than 10 dB, after which their performance is similar. Hence, it is reasonable to take $k = 8$ in the architecture for the case where $M=1$. We observed that with the addition of co-channel interference, stopping earlier than was ascertained in the interference-free case could be beneficial, however we did not change the number of epochs for which the training was performed. This is because in a practical scenario, where we do not have access to the noiseless received signal, we cannot ascertain when to stop, it has to be determined beforehand. Even without dynamically adapting early stopping, the deep channel estimator with $k=8$ beats MMSE estimator up to 10 dB, which also means that it has better interference mitigation capability.

\subsubsection{Single-cell Massive MIMO}
The deep channel estimator is mainly intended for multiple antennas in this paper. Thus, $M$ is set to $64$ and a $128\times64\times64$ matrix is obtained by concatenating the real and imaginary part of the signal with the antennas in the spatial axis. Here, the spatial domain is used to stack up the real and imaginary domain, because this axis is more appropriate for uncorrelated samples in the architecture. In this case, we first observe the impact of an increased number of pilots by varying $N_{p} = 1,2,4$ under the assumption of block type pilot arrangement. Accordingly, the pilots are transmitted periodically over all the subcarriers assuming that the coherence bandwidth is equal to one subcarrier without any loss of generality. The results are shown in Fig. \ref{fig:np}. Clearly an increase from $N_{p} = 1$ to $N_{p} = 4$ benefits the NMSE, but no benefit is obtained beyond that. This is an artifact of the LTE-EPA model which has very high temporal correlation, and consequently needs very few pilots in the time domain to represent the channel accurately. For the rest of the experiments\footnote{A random pilot arrangement is used instead of a block type pilot arrangement in the case of multi-cell massive MIMO, which refers to that pilot tones are allocated to the subcarriers that belong to different OFDM symbols.}, we adopt $N_{p} = 1$. Notice that although there is a single OFDM pilot symbol, orthogonal pilot sequences can still be formed using the frequency domain thanks to the OFDM. 
\begin{figure} [!t] 
\centering 
\includegraphics [width=5in]{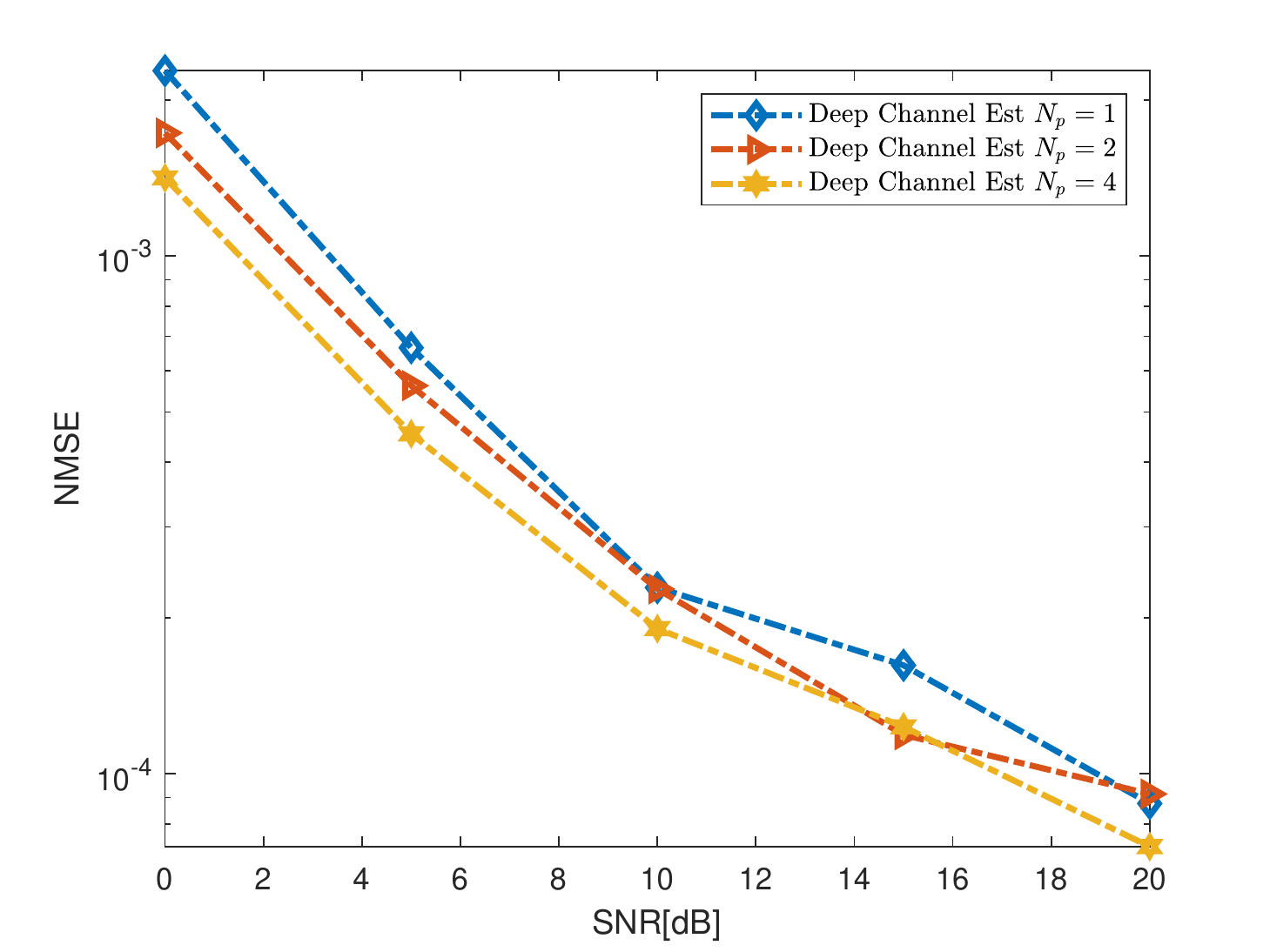}
\caption{NMSE of the proposed estimator for different amount of pilots $N_{p}$.}\label{fig:np}
\end{figure}

Following the same procedure that was adopted for single antenna communication, the optimum number of parameters is determined first, which is tabulated in Table \ref{tab:mhpt} in terms of $k$, the number of epochs and the total number of weights. In what follows, the NMSE as a function of SNR is plotted for different values of $k$. As depicted in Fig \ref{fig:mnik}, the results perfectly reconcile with the single antenna case. That is, at larger noise levels (or lower SNR), smaller values of $k$ perform much better. However at higher SNR, due to early stopping, all the architectures tend to the same NMSE, with the higher $k$ ones performing slightly better. Accordingly, $k = 16$ appears to have the best performance. 
\begin{table}[ht]
\caption{Hyperparameter Tuning for $M=64$}
\label{tab:mhpt}
\begin{center}
\begin{tabular}{|c|c|c|}
\hline
$k$ & Epochs & Total Weight Count\\
\hline
8&4000&1504\\
\hline
16&1970&3776\\
\hline
32&1800&10624\\
\hline
64&1000&33536\\
\hline
\end{tabular}
\end{center}
\end{table}
\begin{figure} [!t] 
\centering 
\includegraphics [width=5in]{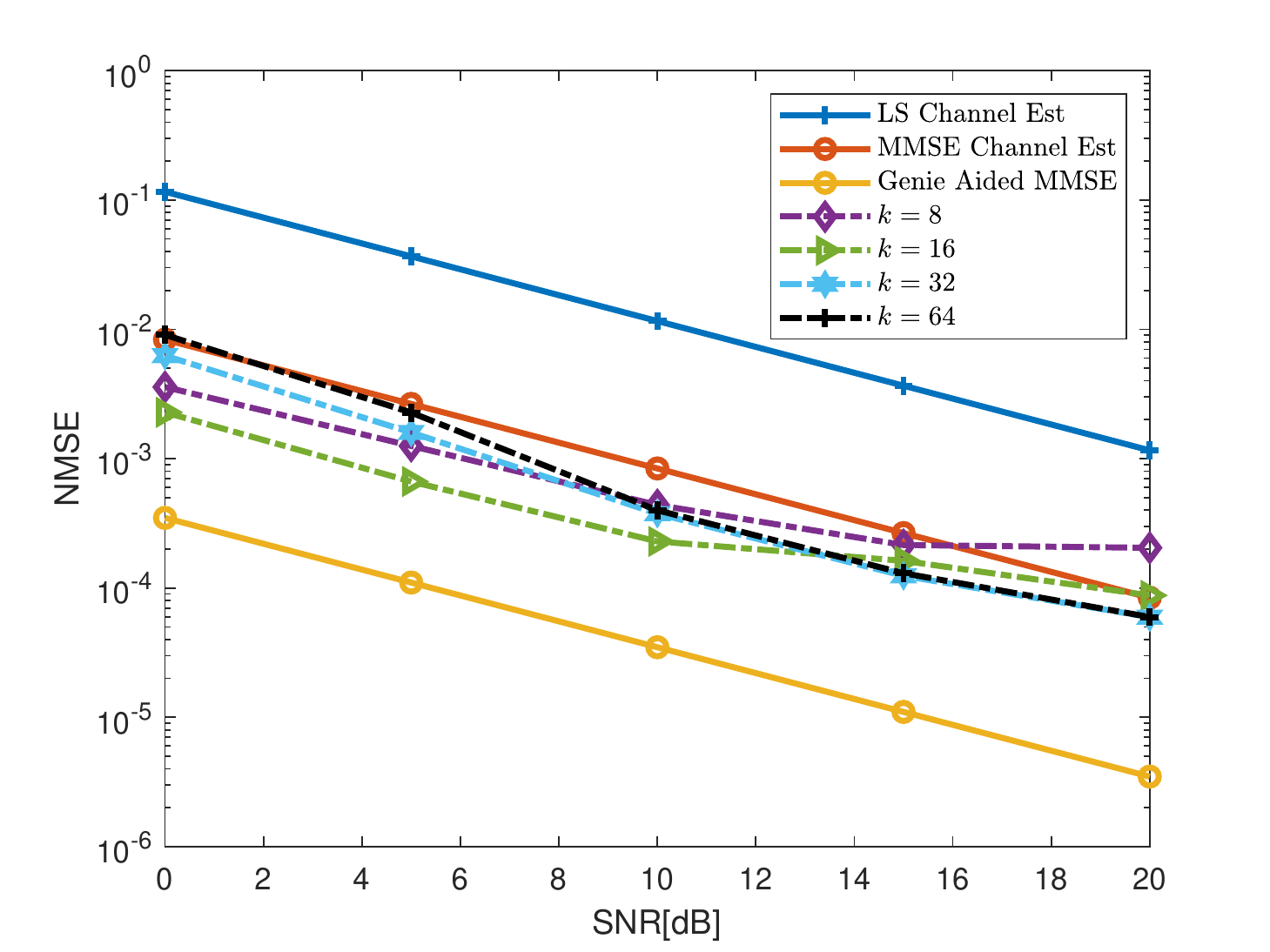}
\caption{NMSE of the proposed estimator for different $k$ and $M=64$ with respect to SNR in comparison to LS and MMSE estimators.}\label{fig:mnik}
\end{figure}

We repeat this experiment for the Kronecker channel model by taking the exponential spatial correlation matrix at the base station with correlation coefficient $\rho=0.5$. However, as shown in Figure. \ref{fig:mk16}, the performance of the deep channel estimator is almost unaffected by the change in spatial correlation of the channel matrix. This is because oversampling is not utilized in this domain. Hence, for brevity the results are only presented for the LTE-EPA model in the rest of the paper.
\begin{figure} [!t] 
\centering 
\includegraphics [width=5in]{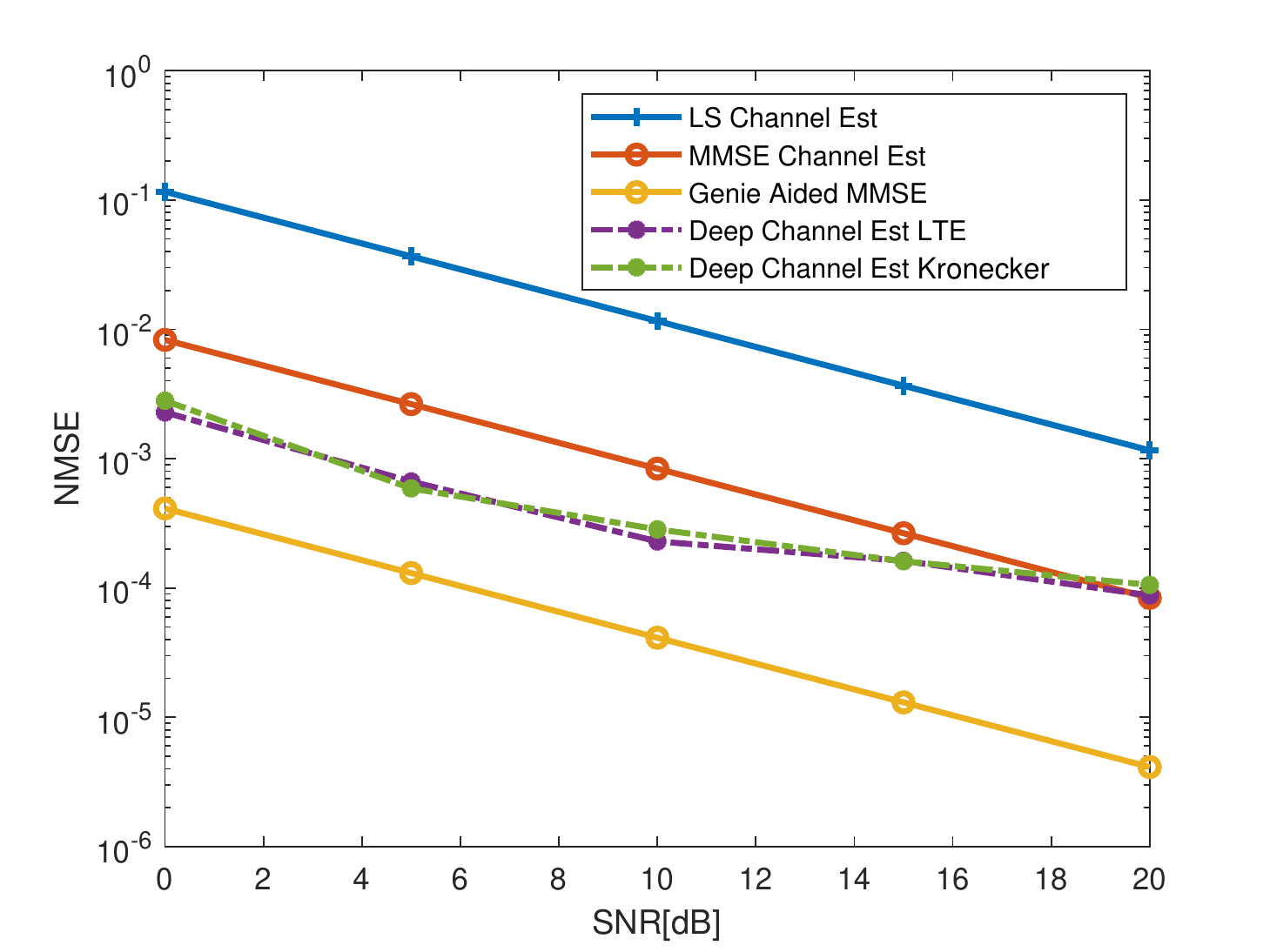}
\caption{NMSE of the proposed estimator for different channel models.}\label{fig:mk16}
\end{figure}

\subsubsection{Multi-cell Massive MIMO} \label{subsec:mcmm}
To assess the robustness of the deep channel estimator against pilot contamination, we first search for the optimum value of $k$, and then exhibit the results. Since base stations allocate random pilots that are spread over the OFDM grid in one coherence time interval, the optimum value of $k$ is searched after contaminating 5\% of the time-frequency grid randomly (but across all antennas) with interference at an SIR of 6 dB. In particular, we checked whether $k=16$ is the optimal architecture as in the case of single cell massive MIMO. 
We found $k = 16$ to outperform all the other architectures, hence the architecture is optimized with $k=16$ for the rest of the multi cell massive MIMO experiments. This experiment is extended by also contaminating 10\% of the OFDM grid for $k=16$. The results for both 5\% and 10\% contamination are presented in Fig. \ref{fig:mik}. In this case, the deep channel estimator outperforms MMSE estimator up to an SNR of 7 dB even in the presence of up to 10\% pilot contamination. The flattening out of the NMSE curve with increased interference is due to not patching the signal in the areas corrupted by interference beyond a certain limit. In the image processing, this corresponds to an upper bound on the size of patches that can be recovered by region inpainting.
\begin{figure*}[!t]
\centering
\subfigure[Random pilot contaminations at SIR = 6 dB. Reference curves are in the presence of 10\% interference at SIR = 6 dB.]{
\label{fig:mik}
\includegraphics[width=5in]{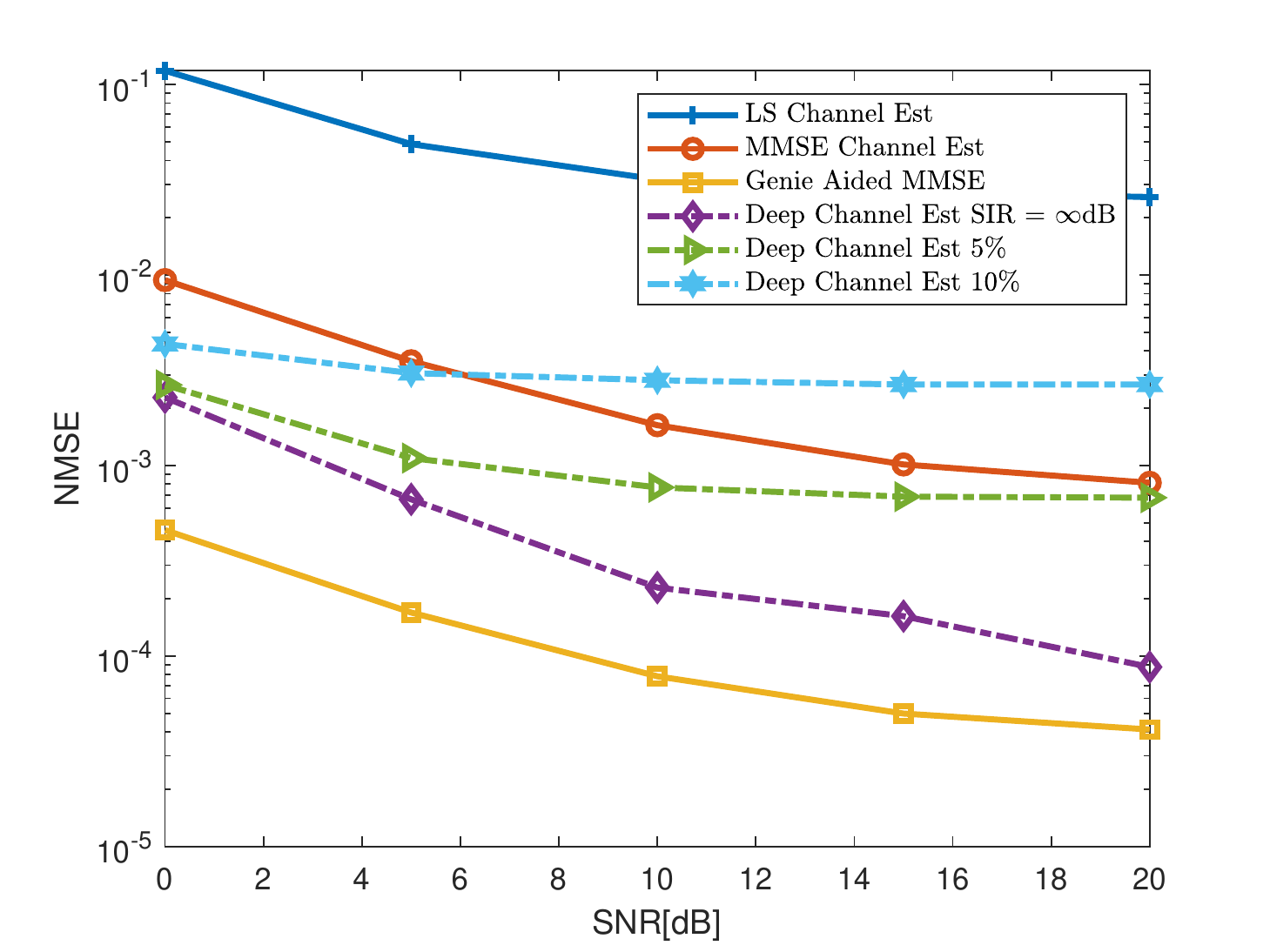}}
\qquad
\subfigure[Contiguous pilot contamination. Reference curves are at SIR = 10 dB.]{
\label{fig:mk16ci}
\includegraphics[width=5in]{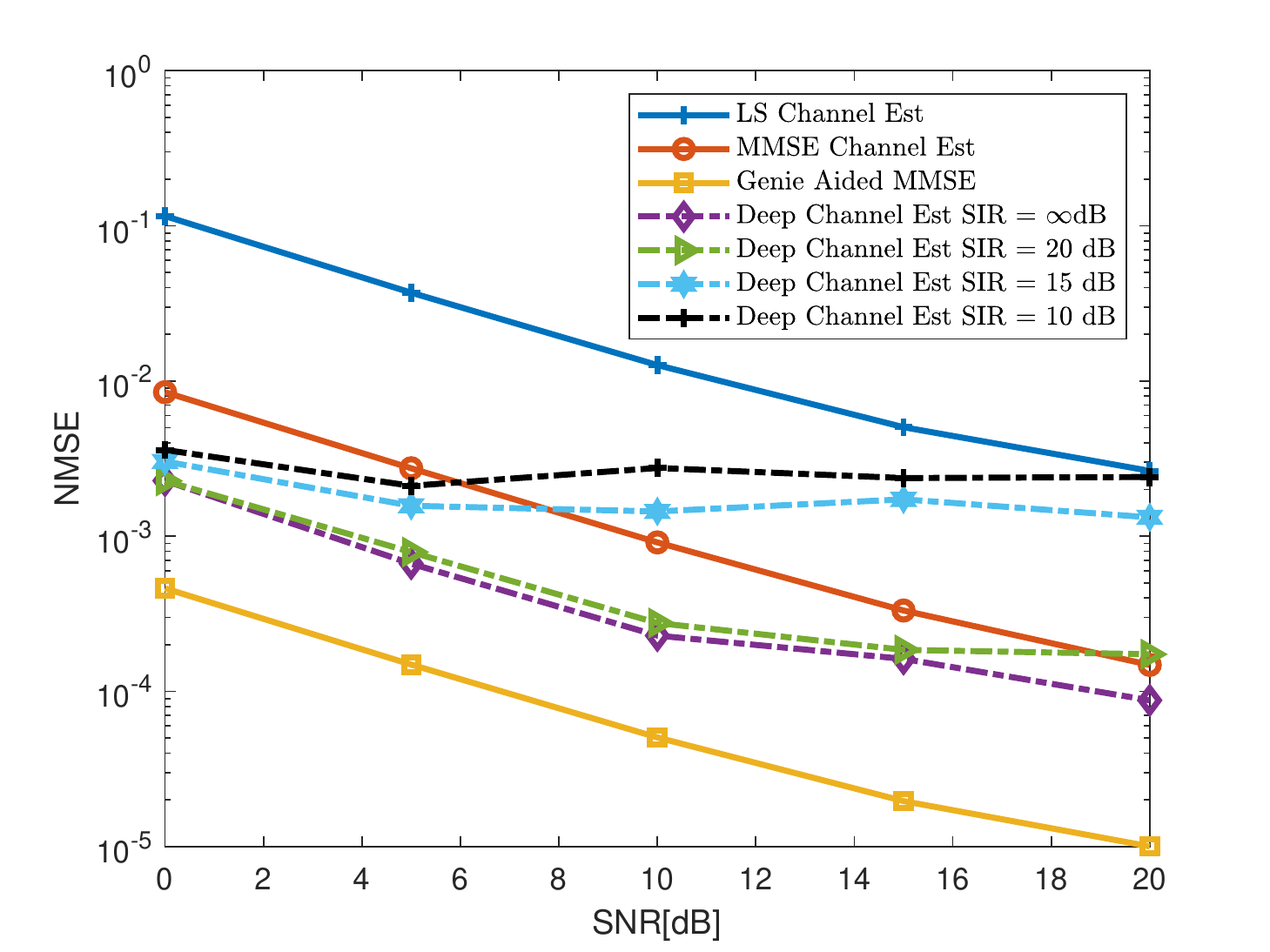}}
\caption{NMSE of the proposed estimator for $k=16$ and $M=64$ with respect to SNR in comparison to LS and MMSE estimators.}
\end{figure*}

To further quantify the pilot contamination performance of our estimator, we verify its robustness for a different power allocation method. Accordingly, pilots are not only randomly but also contiguously distributed over the resource elements. To be more precise, 2 blocks of $8\times8$ squares (corresponding to $\sim3\%$  of the overall time-frequency grid) are chosen randomly, in which interference at SIRs of 10, 15 and 20dB is injected. Although the deep channel estimator in this case can tolerate lower powers of interference than the previous case, its performance, as illustrated in Fig. \ref{fig:mk16ci}, is still better than LS estimator for all SNRs and MMSE estimator up to an SNR of 6 dB for the SIRs that are greater than 10 dB. 

\subsection{Complexity} \label{sec:disc}
Regarding the complexity of the deep channel estimator, it is important to note the trade-off between the number of parameters and the number of epochs required. As seen in Table \ref{tab:hpt}, a decoder of decreased complexity requires a larger number of epochs to attain the minimum NMSE. For instance, in the case of a single antenna the number of parameters for $k = 8$ are 496, in which the NMSE are the least, but it requires 2000 epochs to attain this NMSE. On the other hand, the $k = 64$ architecture has 25,472 parameters, and attains a slightly higher NMSE than $k=8$, but requires a mere 250 epochs to attain this NMSE. This result has in fact been proven for the case of supervised learning of a single hidden layer neural network in \cite{li2019gradient}, where they show that as the degree of overparameterization of the NN increases, it takes fewer epochs to converge to one of the many global minima in its objective function's landscape. As a result, if the deep channel estimator was to be deployed in a latency critical application and subject to online training, where a slightly higher NMSE could be tolerated, one should use the model with a higher $k$ value. 

For the case of $64$ antennas, the optimal architecture surfaces for $k=16$, which has only 3776 parameters but requires 1970 epochs to attain its lowest NMSE. On the other hand, $k=128$ has 116,224 parameters, but requires only 1000 epochs to attain its lowest achievable NMSE, which is much higher than $k=16$. For training around the same number of epochs such as 2000, the single antenna architecture has 496 parameters, while the massive MIMO architecture has 3776 parameters. This comparison is quite important, and specifies the sub-linear increase in computational complexity with the number of antennas. 

\section{Conclusions} \label{Conclusions}
In this paper we proposed a novel deep channel estimator comprised of a DNN followed by a simple LS-type estimator.   This deep channel estimator exhibits superior performance compared to LS and MMSE estimators that have no inherent way of dealing with pilot contamination (or co-channel interference).  Promisingly, our low-complexity estimator performs better than more complex MMSE estimator, in which the channel correlation matrices are estimated from the samples, and even approaches the ``Genie-Aided MMSE" where the channel statistics are perfectly known for free.  The deep channel estimator appears to exploit correlations in the time-frequency grid very efficiently.  The strong performance is also explained by a supporting mathematical analysis. The salient features of the proposed estimator are as follows. The number of parameters scale at a rate less than the square root of number of antennas, which yields hundreds or thousands of weights as opposed to millions of parameters in conventional DNNs. Furthermore, the proposed estimator is appropriate for any environment or channel type, since it only needs the received signal and some pilots.

It would be interesting as future work to study the deep channel estimator for high mobility channels. Similarly, observing the performance of the deep channel estimator for mmWave channels seems intriguing. Furthermore, enhancing its interference mitigation capability can also be a good future research direction. In particular, some other dictionary learning algorithms can be adapted to our model. Additionally, it would be interesting to observe how our estimator performs when the eigenspace of the covariance matrices of interfering users does not fully overlap with the target user.


\begin{thebibliography}{1} 

\bibitem{Marzetta10}
T. L. Marzetta, 
``Noncooperative cellular wireless with unlimited numbers of base station antennas,"
\emph{IEEE Trans. Wireless Commun.}, vol. 9, no. 11, pp. 3590 - 3600, November 2010.

\bibitem{Goodfellow-et-al-2016}
I. Goodfellow, Y. Bengio, and A. Courville,
\textit{Deep Learning. MIT Press}, 2016, http://www.deeplearningbook.org.

\bibitem{Ulyanov18}
D. Ulyanov, A. Vedaldi, and V. Lempitsky, 
``Deep image prior,"
\emph{in Proc IEEE Conference on Computer Vision and Pattern Recognition}, pp. 9446-9454, June 2018.

\bibitem{Beek95}
J.-J. van de Beek, O. Edfors, M. Sandell, S. Wilson, and P. Borjesson, 
``On channel estimation in OFDM systems," 
\emph{in Proc. IEEE VTC}, vol. 2, pp. 815-819, March 1995.

\bibitem{Ngo13a}
H. Q. Ngo, E. G. Larsson, and T. L. Marzetta, 
``The multicell multiuser MIMO uplink with very large antenna arrays and a finite-dimensional channel," 
\emph{IEEE Trans. Commun.}, vol. 6, no. 61, pp. 2350 - 2361, June 2013.

\bibitem{Ngo13b}
H. Q. Ngo, E. G. Larsson, and T. L. Marzetta,
``Energy and spectral efficiency of very large multiuser MIMO systems,"
\emph{IEEE Trans. Commun.}, vol. 4, no. 61, pp. 1436-1449, April 2013.

\bibitem{Adhikary17}
A. Adhikary, A. Ashikhmin, and T. L. Marzetta, 
``Uplink interference reduction in large scale antenna systems,"
\emph{IEEE Trans. Commun.}, vol. 5, no. 65, pp. 2194-2206, May 2017.

\bibitem{AdhikaryCaire13}
A. Adhikary, J. Nam, J.-Y. Ahn, and G. Caire, 
``Joint spatial division and multiplexing-the large-scale array regime,"
\emph{IEEE Trans. on Info. Theory}, vol. 59, no. 10, pp. 6441-6463, October 2013.

\bibitem{YinGesbert13}
H. Yin, D. Gesbert, M. Filippou, and Y. Liu, 
``A coordinated approach to channel estimation in large-scale multiple-antenna systems," 
\emph{IEEE Journal on Sel. Areas in Communications}, vol. 31, no. 2, pp. 264-273, February 2013.

\bibitem{YouSwindlehurst16}
L. You, X. Gao, A. L. Swindlehurst, and W. Zhong, 
``Channel acquisition for massive MIMO-OFDM with adjustable phase shift pilots," 
\emph{IEEE Trans. on Signal Processing}, vol. 64, no. 6, pp. 1461-1476, March 2016.

\bibitem{Wen15}
C.-K. Wen, K.-K. W. S. Jin, J.-C. Chen, and P. Ting, 
``Channel estimation for massive MIMO using Gaussian-Mixture Bayesian learning,"  
\emph{IEEE Trans. Wireless Commun.}, vol. 14, no. 3, pp. 1356-1368, March 2015.

\bibitem{YinHe16}
H. Yin, L. Cottatellucci, D. Gesbert, R. R. Muller, and G. He, 
``Robust pilot decontamination based on joint angle and power domain discrimination," 
\emph{IEEE Trans. on Signal Processing}, vol. 64, no. 11, p. 2990-3003, June 2016.

\bibitem{Ngo12}
H. Q. Ngo and E. G. Larsson, 
``EVD-based channel estimations for multicell multiuser MIMO with very large antenna arrays," 
\emph{in Proc. IEEE ICASSP}, pp. 3249-3252, March 2012.

\bibitem{Muller14}
R. Muller, L. Cottatellucci, and M. Vehkapera, 
``Blind pilot decontamination," 
\emph{IEEE Journal on Sel. Topics Signal Process.}, vol. 8, no. 5, pp. 773-786, October 2014.

\bibitem{Oshea17}
T. O`Shea and J. Hoydis, 
``An introduction to deep learning for the physical layer," 
\emph{IEEE Transactions on Cognitive Communications and Networking}, vol. 3, no. 4, pp. 563-575, December 2017.

\bibitem{DornerBrink18}
S. Dorner, S. Cammerer, J. Hoydis, and S. ten Brink, 
``Deep learning based communication over the air,"  
\emph{IEEE Journal of Selected Topics in Signal Processing}, vol. 12, no. 1, pp. 132-143, February 2018.

\bibitem{BaleviGitlin17}
E. Balevi and R. D. Gitlin,
``Unsupervised machine learning in 5G networks for low latency communications," 
\emph{in Proc. IEEE IPCCC}, December 2017.

\bibitem{YeLi18}
H. Ye, G. Y. Li, and B.-H. Juang, 
``Power of deep learning for channel estimation and signal detection in OFDM systems," 
\emph{IEEE Wireless Communications Letters}, vol. 7, no. 1, pp. 114-117, February 2018.

\bibitem{BalAnd19}
E. Balevi and J. G. Andrews, 
``One-bit OFDM receivers via deep learning," 
\emph{IEEE Trans. on Communications}, vol. 67, no. 6, pp. 4326 - 4336, June, 2019.

\bibitem{BalAndAntenna19}
E. Balevi and J. G. Andrews, 
``Online antenna tuning in heterogeneous cellular networks with deep reinforcement learning," 
\emph{arXiv preprint arXiv:1903.06787}, 2019.

\bibitem{Heckel18}
R. Heckel and P. Hand, 
``Deep decoder: Concise image representations from untrained non-convolutional networks,"
\emph{arXiv preprint arXiv:1810.03982}, February 2019.

\bibitem{BalAndDCE}
E. Balevi and J. G. Andrews, 
``Deep learning-based channel estimation for high-dimensional signals,"
\emph{arXiv preprint arXiv:1904.09346}, 2019.

\bibitem{GlorotBengio10}
X. Glorot and Y. Bengio, 
``Understanding the difficulty of training deep feedforward neural networks,"
\emph{in Proc. NIPS}, May 2010.

\bibitem{DuLee18}
S. S. Du and J. D. Lee, 
``On the power of over-parametrization in neural networks with quadratic activation,"
\emph{arXiv preprint arXiv:1803.01206}, June 2018.

\bibitem{NIPS2008_3448}
J. Mairal, J. Ponce, G. Sapiro, A. Zisserman, and F. R. Bach, 
``Supervised dictionary learning,"
\emph{in Advances in Neural Information Processing Systems}, pp. 1033–1040, 2009.

\bibitem{paszke2017pytorch}
A. Paszke, S. Gross, S. Chintala, and G. Chanan, 
``PyTorch: Tensors and dynamic neural networks in Python with strong GPU acceleration," vol. 6, 2017.

\bibitem{Loyka01}
S. L. Loyka, 
``Channel capacity of MIMO architecture using the exponential correlation matrix,"
\emph{IEEE Commun. Lett}, vol. 5, no. 9, pp. 369–371, September 2001.

\bibitem{tse2005fundamentals}
D. Tse and P. Viswanath, 
\textit{Fundamentals of wireless communication}. Cambridge University Press, 2005.

\bibitem{kingma2014adam}
D. P. Kingma and J. Ba, 
``Adam: A method for stochastic optimization,"
\textit{arXiv preprint arXiv:1412.6980}, 2014.

\bibitem{li2019gradient}
M. Li, M. Soltanolkotabi, and S. Oymak, 
``Gradient descent with early stopping is provably robust to label noise for overparameterized neural networks," 
\textit{arXiv preprint arXiv:1903.11680}, 2019.


\end{thebibliography}
\end{document}